\documentclass[11pt,reqno]{amsart}

\usepackage{amssymb}
\usepackage{url}
\usepackage{hyperref}
\usepackage[dvips]{color}

\usepackage[all]{xy}
\SelectTips{cm}{} 

\allowdisplaybreaks[4]

\usepackage{mathrsfs}
\let\mathcal\mathscr
\usepackage[mathscr]{eucal}

\newcommand*{\pd}[2]{\mathchoice{\frac{\partial#1}{\partial#2}}
  {\partial#1/\partial#2}{\partial#1/\partial#2}
  {\partial#1/\partial#2}}

\let\phi=\varphi
\let\kappa=\varkappa
\let\epsilon=\varepsilon
\DeclareMathOperator{\rank}{rank}
\DeclareMathOperator{\sym}{sym}

\newcommand*{\eval}[1]{\left.#1\right|}
\newcommand*{\abs}[1]{\left|#1\right|}
\newcommand*{\Ev}{\mathbf{E}}

\theoremstyle{theorem}
\newtheorem{proposition}{Proposition}

\newtheorem{theorem}{Theorem}
\newtheorem{lemma}{Lemma}

\theoremstyle{remark}
\newtheorem{remark}{Remark}

\usepackage{mathrsfs}
\let\mathcal\mathscr
\usepackage[mathscr]{eucal}

\newcommand{\cprime}{\/{\mathsurround=0pt$'$}}

\author{I.S.\,Krasil{\cprime}shchik} \address{Trapeznikov Institute of Control
  Sciences, 65 Profsoyuznaya street, Moscow 117997, Russia}
\email{josephkra@gmail.com} \author{P.~Voj{\v{c}}{\'{a}}k}
\address{Mathematical Institute, Silesian University in Opava, Na
  Rybn\'{\i}\v{c}ku 1, 746 01 Opava, Czech Republic}
\email{Petr.Vojcak@math.slu.cz} \title[4D MASh equation]{On the algebra of
  nonlocal symmetries for the 4D Mart\'{\i}nez Alonso-Shabat equation}

\begin{document}

\begin{abstract}
  We consider the 4D Mart\'{\i}nez Alonso-Shabat equation $\mathcal{E}$
  $u_{ty} = u_z u_{xy} - u_y u_{xz}$ (also referred to as the universal
  hierarchy equation) and using its known Lax pair construct two
  infinite-dimensional differential coverings over $\mathcal{E}$. In these
  coverings, we give a complete description of the Lie algebras of nonlocal
  symmetries. In particular, our results generalize the ones obtained
  in~\cite{Mor-Ser} and contain the constructed there infinite hierarchy of
  commuting symmetries as a subalgebra in a much bigger Lie algebra. 
\end{abstract}

\maketitle

\tableofcontents

\section*{Introduction}
\label{sec:introduction}
To the best of our knowledge, the equation
\begin{equation}
  \label{eq:1}
  u_{ty} = u_z u_{xy} - u_y u_{xz} 
\end{equation}
was introduced in the work~\cite{MA-Sh-2004} by L.~Mart\'{\i}nez Alonso,
A.B.~Shabat, where the authors studied multi-dimensional systems whose
reductions lead to the known $(1+1)$-integrable equations (see
also~\cite{MA-Sh-2002} for additional motivations). By this reason we call
Equation~\eqref{eq:1} the \emph{Mart\'{\i}nez Alonso-Shabat} equation, or
shortly the $4$D \emph{MASh equation}. The equation arises also in
classification of integrable $4$D systems, see~\cite{DFKN}.

A differential covering (Lax pair) with a non-removable parameter was
constructed in~\cite{Mor-2014}, as well as a recursion operator for symmetries
of the $4$D MASh equation. Using this covering, the authors of~\cite{Mor-Ser}
found a hierarchy of nonlocal symmetries and proved its commutativity.

We study Equation~\eqref{eq:1} using the approach successfully applied to
integrable linearly degenerate $3$D systems in~\cite{comp-study}
and~\cite{VerWE}. Expanding defining equations of the Lax pair in formal
series of the spectral parameter, we construct two differential coverings
(which we call the negative and positive ones) and describe the algebras of
nonlocal symmetries in these coverings. As the reader will see, the structure
of these algebras is quite complicated. The commutative hierarchy found
in~\cite{Mor-Ser} appears as a subalgebra in one of them. We also analyze the
action of the recursion operator from~\cite{Mor-2014} on our symmetries.

The structure of the paper is as follows: in Section~\ref{sec:preliminaries}
we present very briefly necessary facts from the geometrical theory of
PDEs~\cite{AMS-book} and differential
coverings~\cite{KV-Trends}. Section~\ref{sec:equat-its-cover} contains the
construction of the positive and negative coverings and defining equtions for
symmetries in them. In Section~\ref{sec:algebr-nonl-symm}, the symmetry
algebras are described.

\section{Preliminaries}
\label{sec:preliminaries}
Let us very shortly recall the necessary theoretical background. All the
details may be found, e.g., in~\cite{AMS-book} and~\cite{KV-Trends}. A
particular implementation of all the general constructions will be presented
in Section~\ref{sec:equat-its-cover}.\medskip

\textbf{Equations.} From the geometrical viewpoint, a differential equation is
a submanifold in a jet space. More precisely, this means the following. Let
$\pi\colon E\to M$ be a locally trivial vector bundle over a smooth manifold,
$\dim M = n$, $\rank\pi = m$, and $\pi_\infty\colon J^\infty(\pi)\to M$ be the
corresponding bundle of infinite jets. For us, a differential equation
(imposed on sections of~$\pi$) is a
submanifold~$\mathcal{E}\subset J^\infty(\pi)$ obtained by the prolongation
procedure from a submanifold in the space of finite jets. We use the same
notation~$\pi_\infty$ for the restriction
$\eval{\pi_\infty}_{\mathcal{E}}\colon\mathcal{E}\to M$. The structure of
equation on~$\mathcal{E}$ is defined by the Cartan connection~$\mathcal{C}$,
which takes vector fields~$X\in D(M)$ to vector fields
$\mathcal{C}_X\in D(\mathcal{E})$ on~$\mathcal{E}$. The connection is flat,
i.e., $\mathcal{C}_{[X,Y]} = [\mathcal{C}_X,\mathcal{C}_Y]$. The corresponding
integrable $\pi_\infty$-horizontal distribution is called the Cartan
distribution on~$\mathcal{E}$ and its maximal ($n$-dimensional) integral
manifolds are identified with solutions of~$\mathcal{E}$.\medskip

\textbf{Local symmetries.} An (infinitesimal higher local) symmetry
of~$\mathcal{E}$ is a $\pi_\infty$-vertical vector field~$S\in D(\mathcal{E})$
on~$\mathcal{E}$ such that the commutator~$[S,\mathcal{C}_X]$ lies in the
Cartan distribution for any~$X\in D(M)$. Symmetries form a Lie
$\mathbb{R}$-algebra denoted by~$\sym\mathcal{E}$.

To describe~$\sym\mathcal{E}$, consider another vector bundle~$\xi\colon G\to
M$, $\rank\xi = r$, and assume that $\mathcal{E} = \{F = 0\}$ is given as the
set of zeros of some section~$F\in P = \Gamma(\pi_\infty^*(\xi))$, where
$\Gamma(\cdot)$ denotes the $C^\infty(M)$-module of sections. Consider also
the module~$\kappa = \Gamma(\pi_\infty^*(\pi))$ and the linearization operator
\begin{equation*}
  \ell_{\mathcal{E}} = \eval{\ell_F}_{\mathcal{E}}\colon\kappa\to P.
\end{equation*}
Then one has
\begin{equation*}
  \sym\mathcal{E} = \ker\ell_{\mathcal{E}}.
\end{equation*}
Thus, to any symmetry~$S\in\sym\mathcal{E}$ there corresponds a
section~$\phi\in \kappa$, its generating section, or characteristic, and we
use the notation~$S = \Ev_\phi$ in this case. The commutator of symmetries
generates a bracket in the $\mathbb{R}$-space of generating sections defined
by $[\Ev_\phi,\Ev_\psi] = \Ev_{\{\phi,\psi\}}$. The bracket
$\{\cdot\,,\cdot\}$ is called the (higher) Jacobi bracket.\medskip

\textbf{Differential coverings.} Let
$\tau\colon\tilde{\mathcal{E}}\to\mathcal{E}$ be a locally trivial bundle. It
is called a (differential) covering over~$\mathcal{E}$ if there exists a flat
connection~$\tilde{\mathcal{C}}$ in the bundle $\tilde{\pi}_\infty =
\pi_\infty\circ\tau\colon\tilde{\mathcal{E}}\to M$ such that
$\tau_*(\tilde{\mathcal{C}}_X) = \mathcal{C}_X$ for any field~$X\in D(M)$. The
manifold~$\tilde{\mathcal{E}}$, locally at least, is always an equation in
some bundle over~$M$ and is called the covering equation. The number~$s =
\rank\tau$ is the covering dimension and it may be infinite. Coordinates in
fibers of~$\tau$ are called nonlocal variables.

Let $\tau_i\colon\mathcal{E}_i\to\mathcal{E}$, $i=1$, $2$, be two
coverings. Then their Whitney product~$\tau_1\oplus\tau_2$ carries a natural
structure of a covering called the Whitney product of~$\tau_1$
and~$\tau_2$ and all the arrows in the diagram
\begin{equation*}
  \xymatrixcolsep{5pc}
  \xymatrixrowsep{3pc}
  \xymatrix{
    \mathcal{E}_1\times_{\mathcal{E}}\mathcal{E}_2
    \ar[r]^-{\tau_2^*(\tau_1)}
    \ar[d]_-{\tau_1^*(\tau_2)}
    \ar[dr]^-{\tau_1\oplus\tau_2}&\mathcal{E}_1
    \ar[d]^-{\tau_1}\\
    \mathcal{E}_2\ar[r]^-{\tau_2}&\mathcal{E}
  }
\end{equation*}
are coverings.

A B\"{a}cklund transformation between equations~$\mathcal{E}_1$
and~$\mathcal{E}_2$ is a diagram of the form
\begin{equation*}
  \xymatrix{
    &\tilde{\mathcal{E}}\ar[dl]_-{\tau_1}\ar[dr]^-{\tau_2}&\\
    \mathcal{E}_1&&\mathcal{E}_2\rlap{,}
  }
\end{equation*}
where~$\tau_1$ and~$\tau_2$ are coverings. When $\mathcal{E}_1 =
\mathcal{E}_2$ then we speak about B\"{a}cklund auto-transformation. If the
equations $\mathcal{E}_i$, $i=1$, $2$, are given by the systems~$\{F_i(u_i) =
0\}$ then the system $\{\tilde{F}(u_1,u_2) = 0\}$ that corresponds
to~$\tilde{\mathcal{E}}$ possesses the following property: if~$u_1$ is a
solution of~$\mathcal{E}_1$ and~$(u_1,u_2)$ is a solution
of~$\tilde{\mathcal{E}}$ then~$u_2$ solves~$\mathcal{E}_2$ and vice versa.
\medskip

\textbf{Nonlocal symmetries and shadows.} A nonlocal symmetry of~$\mathcal{E}$
in the covering~$\tau$ is a symmetry of~$\tilde{\mathcal{E}}$. These
symmetries form the algebra $\sym_\tau\mathcal{E} =
\sym\tilde{\mathcal{E}}$. Thus, to find nonlocal symmetries, we need to solve
the equation~$\ell_{\tilde{\mathcal{E}}}(\tilde{\phi}) = 0$.

Denote by~$\mathcal{F}$ and by~$\tilde{\mathcal{F}}$ the algebras of smooth
functions on~$\mathcal{E}$ and~$\tilde{\mathcal{E}}$, respectively. The
projection $\tau$ leads to the embedding
$\tau^*\colon \mathcal{F} \to \tilde{\mathcal{F}}$. We say that an
$\mathbb{R}$-linear derivation $Y\colon \mathcal{F} \to \tilde{\mathcal{F}}$
is a shadow in~$\tau$ if the diagram
\begin{equation*}
  \xymatrix{
    \mathcal{F}\ar[r]^{\mathcal{C}_X}\ar[d]_{Y}&\mathcal{F}\ar[d]^{Y}\\
    \tilde{\mathcal{F}}\ar[r]^{\tilde{\mathcal{C}}_X}&\tilde{\mathcal{F}}
  }
\end{equation*}
is commutative for any~$X\in D(M)$. In particular, for any nonlocal symmetry
$\tilde{S}\colon \tilde{\mathcal{F}}\to \tilde{\mathcal{F}}$ the restriction
$\eval{\tilde{S}}_{\mathcal{F}}\colon \mathcal{F}\to\tilde{\mathcal{F}}$ is a
shadow. We say that~$\tilde{S}$ is invisible if its shadow vanishes. Note that
any local symmetry~$S$ may be regarded as a $\tau$-shadow if one takes the
composition~$\tau^*\circ S$. A nonlocal symmetry~$\tilde{S}$ is a lift of a
shadow~$Y$ if $\eval{\tilde{S}}_{\mathcal{F}} = Y$. A lift, if it exists, is
defined up to invisible symmetries. The defining equation for shadows is
\begin{equation*}
  \tilde{\ell}_{\mathcal{E}}(\tilde{\phi}) = 0,
\end{equation*}
where $\tilde{\ell}_{\mathcal{E}}$ is the natural extension of the
linearization operator from~$\mathcal{E}$ to~$\tilde{\mathcal{E}}$.\medskip

\textbf{Recursion operators (see~\cite{KVV-Springer,Marvan-another-look}).}
Let an equation~$\mathcal{E}$ be given by~$\{F(u) = 0\}$. Then its tangent
equation is
\begin{equation*}
  \mathcal{TE}:\
  \begin{array}{l}
    F(u) = 0\\ \ell_F(p) = 0,
  \end{array}
\end{equation*}
where~$p=(p^1,\dots,p^m)$ is a new unknown of the same dimension as~$u$. The
projection $\mathbf{t}\colon \mathcal{TE}\to\mathcal{E}$, $(u,p)\mapsto u$, is
called the tangent covering of~$\mathcal{E}$. Properties of~$\mathcal{TE}$ are
closely related with symmetries of~$\mathcal{E}$: sections
of~$\mathcal{}\mathbf{t}$ which take the Cartan distribution on~$\mathcal{E}$
to that on~$\mathcal{TE}$ are in one-to-one correspondence with symmetries.

Let~$\mathcal{R}$ be a B\"{a}cklund auto-transformation of
$\mathcal{TE}$. Then it relates shadows of symmetries of~$\mathcal{E}$ with
each other, i.e., may be understood as recursion operator. 

\section{The equation and its coverings}
\label{sec:equat-its-cover}

Here we present the necessary formulas for the computations to be done in
Section~\ref{sec:algebr-nonl-symm}.\medskip

\textbf{Internal coordinates and the total derivatives.} The
manifold~$\mathcal{E}$ corresponding to Equation~\eqref{eq:1} lies
in~$J^\infty(\pi)$, where
$\pi\colon \mathbb{R}\times\mathbb{R}^4 \to \mathbb{R}^4$ is the trivial
bundle. We denote the coordinates in the base by~$x$, $y$, $z$, $t$, while~$u$
denotes a coordinate in the fiber. Then internal coordinates
\begin{equation*}
  u_{x^iz^j},\quad u_{x^iz^jy^k},\quad u_{x^iz^jt^l},\qquad i,j\geq0,\ k,l>0,
\end{equation*}
on~$\mathcal{E}$ arise. Then the Cartan connection is completely determined by
its values on the basis vector fields $\pd{}{x}$, $\pd{}{y}$, $\pd{}{z}$,
$\pd{}{t}$. The result is the corresponding total derivatives
on~$\mathcal{E}$:
\begin{align*}
  D_x&=\pd{}{x} + \sum_{i,j\geq0,k,l>0}\left(u_{x^{i+1}z^j}\pd{}{u_{x^iz^j}} +
       u_{x^{i+1}z^jy^k}\pd{}{u_{x^iz^jy^k}} +
       u_{x^{i+1}z^jt^l}\pd{}{u_{x^iz^jt^l}}\right),\\
  D_y&=\pd{}{y} + \sum_{i,j\geq0,k,l>0}\left(u_{x^iz^jy}\pd{}{u_{x^iz^j}} +
       u_{x^iz^jy^{k+1}}\pd{}{u_{x^iz^jy^k}} +
       D_x^iD_y^jD_t^{l-1}(u_zu_{xy}-u_yu_{xz})\pd{}{u_{x^iz^jt^l}}\right),\\
  D_z&=\pd{}{z} + \sum_{i,j\geq0,k,l>0}\left(u_{x^iz^{j+1}}\pd{}{u_{x^iz^j}} +
       u_{x^iz^{j+1}y^k}\pd{}{u_{x^iz^jy^k}} +
       u_{x^iz^{j+1}t^l}\pd{}{u_{x^iz^jt^l}}\right),\\
  D_t&=\pd{}{t} + \sum_{i,j\geq0,k,l>0}\left(u_{x^iz^jt}\pd{}{u_{x^iz^j}} +
       D_x^iD_y^jD_y^{k-1}(u_zu_{xy}-u_yu_{xz})\pd{}{u_{x^iz^jy^k}} +
       u_{x^iz^jt^{l+1}}\pd{}{u_{x^iz^jt^l}}\right).
\end{align*}
The Cartan distribution on~$\mathcal{E}$ is spanned by these fields.
\medskip

\textbf{The defining equations for local symmetries.}  The linearization of
Equation~\eqref{eq:1} has the form
\begin{equation}
  \label{eq:2}
  D_yD_t(\phi) =u_{xy}D_z(\phi) - u_{xz}D_y(\phi) + u_zD_xD_y(\phi) -
  u_yD_xD_z(\phi),
\end{equation}
where $\phi$ is a function that depends on a finite number of internal
coordinates. The vector field on~$\mathcal{E}$ that corresponds to a
solution~$\phi$ is
\begin{equation}\label{eq:7}
  \Ev_\phi = \sum_{i,j\geq0,k,l>0}\left(D_x^iD_z^j(\phi)\pd{}{u_{x^iz^j}} +
    D_x^iD_z^jD_y^k(\phi)\pd{}{u_{x^iz^jy^k}} +
    D_x^iD_z^jD_t^l(\phi)\pd{}{u_{x^iz^jt^l}}\right),
\end{equation}
but we shall mainly deal with the generating functions~$\phi$ rather than with
the fields~$\Ev_\phi$ themselves.

Note that it can be easily shown that Equation~\eqref{eq:1} admits point
symmetries only, i.e., solutions of~\eqref{eq:2} may depend only on the
variables~$x$, $y$, $z$, $t$, $u$, $u_x$, $u_y$, $u_z$, and~$u_t$.
\medskip

\textbf{The $\tau^+$- and $\tau^-$ coverings}
All our subsequent nonlocal constructions are based on the covering
\begin{equation}
  \label{eq:3}
  w_t = u_z w_x - \lambda^{-1} w_z ,\quad
  w_y = \lambda u_y w_x,
\end{equation}
where $0\neq\lambda\in\mathbb{R}$ and~$w$ is the nonlocal variable,
see~\cite{Mor-2014}. It is readily checked that the compatibility conditions
for the overdetermined system~\eqref{eq:3} amount to Equation~\eqref{eq:1}. We
denote the covering~\eqref{eq:3} by~$\tau_\lambda$.
\begin{remark}
  \label{sec:rem-1}
  At first glance, the covering~$\tau_\lambda$ is one-dimensional. This is not
  the case, actually, because $x$- and $z$- derivatives of~$w$ are not defined
  in~\eqref{eq:3}. To make the definition complete, we must introduce infinite
  number of nonlocal variables~$w^{\alpha,\beta}$, $\alpha$,
  $\beta=0,1,2,\dots$, $w^{0,0} = w$ and set
  \begin{gather*}
    w_x^{\alpha,\beta} = w^{\alpha+1,\beta},\qquad w_z^{\alpha,\beta} = 
    w^{\alpha,\beta+1}\\
    w_t^{\alpha,\beta} =
    (u_z w_x - \lambda^{-1} w_z)_{x^\alpha z^\beta},\quad
    w_y^{\alpha,\beta} = (\lambda u_y w_x)_{x^\alpha z^\beta}.
  \end{gather*}
  So, \eqref{eq:3} defines an infinite-dimensional covering.
\end{remark}

Assume now that $w = w(\lambda)$ and consider the expansion $w =
\sum_{i\in\mathbb{Z}}\lambda^iw_i$. Substituting the latter into~\eqref{eq:3},
we get
\begin{equation}
  \label{eq:4}
  w_{i,t} = u_zw_{i,x} - w_{i+1,z},\quad w_{i,y} = u_yw_{i-1,x},\qquad
  i\in\mathbb{Z}.
\end{equation}
Thus, we obtain an infinite-dimensional covering over~$\mathcal{E}$, but the
problem is that this is `bad infinity' which has `neither beginning nor
end'. To overcome this inconvenience, we divide~\eqref{eq:4} in two parts
assuming that $w_i = 0$ for $i>0$ in one case and $w_i = 0$ for $i<0$ in the
other. In this way, we obtain two different coverings that we call the
negative ($\tau^-$) and positive ($\tau^+$) ones, respectively. After suitable
relabellings, the defining equations for these coverings acquire the form
\begin{align}
  \label{eq:5}
  \tau^-\colon\mathcal{E}^-\to\mathcal{E}&\quad
                \left|\begin{array}{l}
                  r_0 = y,\\
                  r_{i,t} = u_zu_y^{-1}r_{i-1,y} - r_{i-1,z},\\[2pt]
                  r_{i,x} = u_y^{-1}r_{i-1,y},\quad i \geq 1,
                \end{array}\right.
  \intertext{and}
  \label{eq:6}
  \tau^+\colon\mathcal{E}^+\to\mathcal{E}&\quad
                \left|\begin{array}{l}
                  q_{-1} = x,\ q_0 = u,\\
                  q_{i,y} = u_y q_{i-1,x},\\
                  q_{i,z} = u_z q_{i-1,x} - q_{i-1,t}, \quad i \geq 1.
                \end{array}\right.
\end{align}
So, $q_1$, $q_2\dots$ are the nonlocal variables in~$\tau^+$ and $r_1$,
$r_2,\dots$ are those in~$\tau^-$.

\begin{remark}
  \label{sec:rem-2}
  Strictly speaking, we must enrich~\eqref{eq:5} with infinite number of
  formal variables that would define $y$- and $z$-derivatives of~$r$. In a
  similar way, additional variables that define $x$- and $t$-derivatives
  of~$q$ are needed (cf.~Remark~\ref{sec:rem-1}). To be more precise,
  in~$\tau^-$, we consider the variables~$r_i^{\alpha,\beta}$, $\alpha$,
  $\beta = 0,1,\dots$, such that~$r_i^{0,0} = r_i$ and
  \begin{gather*}
    r_{i,y}^{\alpha,\beta} = r_i^{\alpha+1,\beta},\quad r_{i,z}^{\alpha,\beta} =
    r_i^{\alpha,\beta+1},\\
    r_{i,t}^{\alpha,\beta} = (u_zu_y^{-1}r_{i-1,y} - r_{i-1,z})_{y^\alpha
      z^\beta},\quad r_{i,x}^{\alpha,\beta} = (u_y^{-1}r_{i-1,y})_{y^\alpha
      z^\beta}
  \end{gather*}
  Similarly, we introduce~$q_i^{\alpha,\beta}$ in~$\tau^+$ and set~$q_i^{0,0}
  = q_i$,
  \begin{gather*}
    q_{i,x}^{\alpha,\beta} = q_i^{\alpha+1,\beta},\quad q_{i,t}^{\alpha,\beta}
    = q_i^{\alpha,\beta+1},\\
    q_{i,y}^{\alpha,\beta} = (u_yq_{i-1,x})_{x^\alpha t^\beta},\quad
    q_{i,z}^{\alpha,\beta} = (u_zq_{i-1,x} - q_{i-1,t})_{x^\alpha t^\beta}.
  \end{gather*}
  But, as we shall see below, this formalization does not influence
  the subsequent computations.
\end{remark}
\medskip

\textbf{The defining equations for nonlocal symmetries.}
Let us begin with writing down the total derivatives in the negative and
positive coverings. In~$\tau^-$, due to~\eqref{eq:5} and
Remark~\ref{sec:rem-2}, one has
\begin{equation*}
  D_x^- = D_x + X^-,\quad D_y^- = D_y + Y^-,\quad D_z^- = D_z + Z^-,\quad
  D_t^- = D_t + T^-, 
\end{equation*}
where
\begin{align*}
  X^-
  & = D_x + \sum_{i=1}^\infty\sum_{\alpha,\beta=0}^\infty
    (u_y^{-1}r_{i-1,y})_{y^\alpha
    z^\beta}\pd{}{r_i^{\alpha,\beta}},\\ 
  Y^-
  & = D_y + \sum_{i=1}^\infty\sum_{\alpha,\beta=0}^\infty
    r_i^{\alpha+1,\beta}\pd{}{r_i^{\alpha,\beta}},\\
  Z^-& = D_z + \sum_{i=1}^\infty\sum_{\alpha,\beta=0}^\infty
    r_i^{\alpha,\beta+1}\pd{}{r_i^{\alpha,\beta}},\\
  T^-& = D_t+ \sum_{i=1}^\infty\sum_{\alpha,\beta=0}^\infty
    (u_zu_y^{-1}r_{i-1,y} - r_{i-1,z})\pd{}{r_i^{\alpha,\beta}}.
\end{align*}
The total derivatives in~$\tau^+$ are
\begin{equation*}
  D_x^+ = D_x + X^+,\quad D_y^+ = D_y + Y^+,\quad D_z^+ = D_z + Z^+,\quad
  D_t^+ = D_t + T^+, 
\end{equation*}
where
\begin{align*}
  X^+& = D_x + \sum_{i=1}^\infty\sum_{\alpha,\beta=0}^\infty
  q_i^{\alpha+1,\beta}\pd{}{q_i^{\alpha,\beta}},\\
  Y^+& = D_y + \sum_{i=1}^\infty\sum_{\alpha,\beta=0}^\infty
  (u_yq_{i-1,x})_{x^\alpha t^\beta}\pd{}{q_i^{\alpha,\beta}},\\
  Z^+& = D_z + \sum_{i=1}^\infty\sum_{\alpha,\beta=0}^\infty
  (u_zq_{i-1,x} - q_{i-1,t})_{x^\alpha t^\beta}\pd{}{q_i^{\alpha,\beta}},\\
  T^+& = D_t + \sum_{i=1}^\infty\sum_{\alpha,\beta=0}^\infty
  q_i^{\alpha,\beta+1}\pd{}{q_i^{\alpha,\beta}}. 
\end{align*}
Finally, the total derivatives in the Whitney product $\tau^\pm =
\tau^-\oplus\tau^+$ of~$\tau^-$ and~$\tau^+$ read
\begin{equation*}
  D_x^\pm = D_x + X^- + X^+,\ D_y^\pm = D_y + Y^- + Y^+,\ D_z^\pm = D_z + Z^-
  + Z^+,\ D_t^\pm = D_t + T^- + T^+
\end{equation*}
and the lift of~$\ell_{\mathcal{E}}$ to $\tau^\pm$ will be denoted
by~$\ell_{\mathcal{E}}^\pm$ with the obvious meaning of the notation.

To proceed, let us agree on notation. Denote by~$\Ev_\phi^\pm$ the field
on~$\tau^\pm$ obtained from the field~$\Ev_\phi$ presented in~\eqref{eq:7} by
changing the total derivatives~$D_\bullet$ to~$D_\bullet^\pm$, where $\bullet$
denotes $x$, $y$, $z$ or~$t$. We also obtain
operators~$\ell_{\mathcal{E}}^\pm$ from~$\ell_{\mathcal{E}}$ in the same way.

In this notation, any $\tau^-$-nonlocal symmetry is of the form
\begin{equation*}
  S = \Ev_\phi^- +
  \sum_{i=1}^\infty\sum_{\alpha,\beta=0}^\infty
  \phi_i^{\alpha,\beta}\pd{}{r_i^{\alpha,\beta}},
\end{equation*}
where~$\phi$, $\phi_i^{\alpha,\beta}$ are functions on on~$\tau^-$. Then
$\phi_i^{\alpha,\beta} = (D_y^-)^\alpha(D_z^-)^\beta(\phi^i)$, $\phi^i =
\phi_i^{\alpha,\beta}$, and
\begin{align}\nonumber
  \ell_{\mathcal{E}}^-(\phi)&= 0,\\\label{eq:8}
  D_t^-(\phi^i)&= u_y^{-2}(u_yD_z^-(\phi) - u_zD_y^-(\phi))r_{i-1,y} +
                 u_zu_y^{-1}D_y^-(\phi^{i-1}) - D_z^-(\phi^{i-1}),\\\nonumber
  D_x^-(\phi^i)&= -u_y^{-2}D_y^-(\phi)r_{i-1,y} + u_y^{-1}D_y^-(\phi^{i-1}).
\end{align}
Hence, any such a symmetry $S = S_\Phi$ is completely determined by the
vector-function $\Phi =(\phi,\phi^1,\dots)$ and the formula $[S_\Phi,S_\Psi] =
S_{\{\Phi,\Psi\}}$ defines a bracket on the space of these functions. Nonlocal
shadows are just the functions~$\phi$ that satisfy the first of
Equations~\eqref{eq:8}, while invisible symmetries are $\Phi =
(0,\phi^1,\dots)$ with
\begin{equation}\label{eq:12}
  \begin{array}{l}
   D_t^-(\phi^i)= u_zu_y^{-1}D_y^-(\phi^{i-1}) - D_z^-(\phi^{i-1}),\\[2pt]
  D_x^-(\phi^i)= u_y^{-1}D_y^-(\phi^{i-1}).
  \end{array}
\end{equation}

Of course, the scheme is almost the same in~$\tau^+$. Any symmetry is
\begin{equation*}
  S = \Ev_\phi^+ +
  \sum_{i=1}^\infty\sum_{\alpha,\beta=0}^\infty
  \phi_i^{\alpha,\beta}\pd{}{q_i^{\alpha,\beta}},
\end{equation*}
where~$\phi$, $\phi_i^{\alpha,\beta}$ are functions on on~$\tau^+$ and
$\phi_i^{\alpha,\beta} = (D_x^+)^\alpha(D_t^+)^\beta(\phi^i)$, $\phi^i =
\phi_i^{\alpha,\beta}$. The defining equations for~$\phi$ and~$\phi^i$ are
\begin{align}\nonumber
  \ell_{\mathcal{E}}^+(\phi)&= 0,\\\label{eq:9}
  D_y^+(\phi^i)&= D_y^+(\phi)q_{i-1,x} + u_yD_x^+(\phi^{i-1})\\\nonumber
  D_z^+(\phi^i)&= D_z^+(\phi)q_{i-1,x} + u_zD_x^+(\phi^{i-1}) -D_t^+(\phi^{i-1}).
\end{align}
As above, we introduce generating vector-functions $\Phi =
(\phi,\phi^1,\dots)$ and using the notation $S = S_\Phi$ define the bracket
between these functions. Nonlocal shadows in the positive covering are
identified with solutions of~$\ell_{\mathcal{E}}^+(\phi) = 0$, while invisible
symmetries $\Phi = (0,\phi^1,\dots)$, where $\phi^i$ satisfy the system
\begin{equation}\label{eq:13}
  \begin{array}{l}
  D_y^+(\phi^i)= u_yD_x^+(\phi^{i-1})\\[2pt]
  D_z^+(\phi^i)= u_zD_x^+(\phi^{i-1}) -D_t^+(\phi^{i-1}).
  \end{array}
\end{equation}

Symmetries in the Whitney product are vector fields
\begin{equation*}
  S = \Ev_\phi^\pm +
  \sum_{i=1}^\infty\sum_{\alpha,\beta=0}^\infty\left(
    D_y^\alpha D_z^\beta(\phi_i^-)\pd{}{r_i^{\alpha,\beta}}+
    D_x^\alpha D_t^\beta(\phi_i^+)\pd{}{q_i^{\alpha,\beta}}
  \right),
\end{equation*}
where the functions~$\phi^\pm$, $\phi_i^-$, $\phi_i^+\in
C^\infty(\mathcal{E}^-\times_{\mathcal{E}}\mathcal{E}^+)$ enjoy the relations
\begin{equation}
  \label{eq:15}
  \begin{array}{l}
    \ell_{\mathcal{E}}^\pm(\phi^\pm) = 0,\\[3pt]
    D_t^\pm(\phi^i)= u_y^{-2}(u_yD_z^\pm(\phi) - u_zD_y^\pm(\phi))r_{i-1,y} +
    u_zu_y^{-1}D_y^\pm(\phi^{i-1}) - D_z^\pm(\phi^{i-1}),\\[2pt]
    D_x^\pm(\phi^i)= -u_y^{-2}D_y^\pm(\phi)r_{i-1,y} +
    u_y^{-1}D_y^\pm(\phi^{i-1}),\\[3pt]
    D_y^\pm(\phi^i)= D_y^\pm(\phi)q_{i-1,x} + u_yD_x^\pm(\phi^{i-1})\\[2pt]
    D_z^\pm(\phi^i)= D_z^\pm(\phi)q_{i-1,x} + u_zD_x^\pm(\phi^{i-1}) -
    D_t^\pm(\phi^{i-1}). 
  \end{array}
\end{equation}

\begin{remark}
  \label{sec:rem-3}
  A useful instrument in analysis of Lie algebra structures is the weights
  (gradings) that may be assigned to all the variables in the coverings under
  consideration and all polynomial functions in these variables. Namely, if we
  set the weights on independent variables to be
  \begin{equation*}
    x\mapsto\abs{x},\quad y\mapsto\abs{y},\quad z\mapsto\abs{z}, \quad
    t\mapsto\abs{t},
  \end{equation*}
  then from Equations~\eqref{eq:1}, \eqref{eq:5} and~\eqref{eq:6} it follows
  that
  \begin{equation*}
    \abs{u} = \abs{x} + \abs{z} - \abs{t},\quad \abs{r_i} =\abs{y} + i(\abs{t}
    - \abs{z}),\quad \abs{q_i} = \abs{x} +(i+1)(\abs{z} - \abs{t}).
  \end{equation*}
  To a vector field $A\pd{}{a}$ we assign the weight~$\abs{A} - \abs{a}$. Then
  for any two fields one has~$\abs{[A,B]} = \abs{A} + \abs{B}$. Thus, Lie
  algebras spanned by homogeneous fields become graded.

  So, we have four independent way to introduce weights reflects existence of
  four independent scaling symmetries in~$\sym\mathcal{E}$ (see
  Section~\ref{sec:algebr-nonl-symm}). Weights of differential polynomials are
  computed in an obvious way. In what follows, it will be convenient to use
  the following choice:
  \begin{equation*}
    \abs{x} = -1,\qquad \abs{t} = \abs{y} = \abs{u} = 0,\qquad \abs{z} = 1,
  \end{equation*}
  and thus
  \begin{equation*}
    \abs{r_i} = -i\qquad \abs{q_i} = i.
  \end{equation*}
\end{remark}

To conclude the discussion of structures inherent to the equation under study,
we mention the recursion operator found in~\cite{Mor-2014}. The tangent
equation corresponding to~\eqref{eq:1} is of the form
\begin{align*}
  u_{ty} &= u_z u_{xy} - u_y u_{xz},\\
  p_{yt} &=u_{xy}p_z - u_{zt}p_y + u_zp_{xy} - u_yp_{xz}.
\end{align*}
The B\"{a}cklund transformation\footnote{For the convenience of the subsequent
  exposition, we present it a slightly different from~\cite{Mor-2014} form,
  which is of course equivalent to the original one.}
that relates two copies of~$\mathcal{TE}$ is
\begin{align}
  \label{eq:16}
  \begin{array}{l}
    D_y(\phi)=u_yD_x(\phi')-u_{xy}\phi',\\[2pt]
    D_z(\phi)=-D_t(\phi')+u_zD_x(\phi')-u_{xz}\phi'.
  \end{array}
\end{align}
If $\phi$ is a solution of Equation~\eqref{eq:2} then~$\phi'$ also solves it and
vice versa. The correspondence $\phi\mapsto\phi'$ defined by
relations~\eqref{eq:16} will be denoted by~$\overrightarrow{\mathcal{R}}$ and
the opposite one by~$\overleftarrow{\mathcal{R}}$. The
operator~$\overrightarrow{\mathcal{R}}$ changes the weight by~$+1$,
while~$\overleftarrow{\mathcal{R}}$ changes it by~$-1$.

\section{Algebras of nonlocal symmetries}
\label{sec:algebr-nonl-symm}

We accomplish the construction of the desired algebra in several steps that
are:
\begin{itemize}
\item explicit computation of basic shadows and their lifts to~$\tau^-$,
  $\tau^+$, and $\tau^\pm$ (Proposition~\ref{sec:algebr-nonl-symm-prrop-1});
\item construction of hierarchies by means of commutators of the basic
  symmetries;
\item construction of new hierarchies by somewhat artificial trick
  (Theorems~\ref{sec:algebr-nonl-symm-thm-1}
  and~\ref{sec:algebr-nonl-symm-thm-2});
\item computation of the Lie algebra structure
  (Theorem~\ref{sec:algebr-nonl-symm-thm-3}).
\end{itemize}

\textbf{Notation.} In what follows, $A=A(y,z)$ and~$B=B(x,t)$ are arbitrary
smooth functions. Notation~$S_i^j$ for a symmetry indicates its weight~$i$
(and the position in a hierarchy), while the superscript~$j$ (if any)
enumerates the hierarchies. If a symmetry contains a function~$A$, we compute
its weight assuming~$A=y$; if it contains~$B$, the assumption is~$B=x$.

The coefficient of~$S_i^j$ at~$\pd{}{u}$ (the shadow) will be
denoted by~$s_{i,0}^j$, while its coefficients at~$\pd{}{r_\alpha}$
and~$\pd{}{q_\alpha}$ will be~$s_{i,\alpha}^{j,-}$
and~$s_{i,\alpha}^{j,+}$, respectively. Thus, any symmetry is presented by its
generating vector-function
\begin{equation*}
  S_i^j\sim\left[s_{i,0}^j, s_{i,1}^{j,-}, s_{i,1}^{j,+}, \dots,
    s_{i,\alpha}^{j,-}, s_{i,\alpha}^{j,+},\dots\right],
\end{equation*}
where $s_{i,0}^j$, $s_{i,\alpha}^{j,-}$, $s_{i,\alpha}^{j,+}$ are smooth
functions on $\mathcal{E}^-\otimes_{\mathcal{E}}\mathcal{E}^+$.
\medskip

\textbf{The basic shadows.}
The following shadows are found by direct computations:
\begin{align*}
  \psi_{-1,0}^0&=-u_z,\quad \psi_{0,0}^0=u_t,\quad \psi_{1,0}^0=q_{1,t}-u_tu_x,\\[2pt]
  \omega_{-2,0}^0&=u_y(2r_2+zr_{2,z}-r_{1,y}(r_1+zr_{1,z})),\quad
  \omega_{-1,0}^0=u_y(r_1+zr_{1,z}),\quad
  \omega_{0,0}^0=u-zu_z,\\
  \omega_{1,0}^0&=2q_1-uu_x+zu_t,\quad
  \omega_{2,0}^0=3q_2-2u_xq_1-uq_{1,x}+zq_{1,t}+uu_x^2-zu_tu_x,\\[2pt]
  \xi_{-1,0}(A)&=u_y(A r_{1,y} - A_yr_1+A_zt),\quad
  \xi_{0,0}(A)=-Au_y,\quad \xi_{1,0}(A)=0,\\[2pt]
  \upsilon_{-1,0}(B)&=B,\quad \upsilon_{0,0}(B)=-Bu_x+B_x, u -B_tz,\\
  \upsilon_{1,0}(B)&=B(u_x^2-q_{1,x})+B_x \, (q_1-uu_x) +B_t \,
  zu_x+\frac{1}{2}B_{xx}u^2+\frac{1}{2}B_{tt}z^2-B_{tx}zu.
\end{align*}

\begin{proposition}
  \label{sec:algebr-nonl-symm-prrop-1}
  All the above listed shadows admit lifts to~$\tau^\pm$.
\end{proposition}
\begin{proof}
  The lifts of the shadows~$\psi_{i,0}^j$ and~$\omega_{i,0}^j$ are described
  explicitly. Namely, we set
  \begin{align*}
    &\psi_{-1,\alpha}^{0,-}=-r_{\alpha,z},\qquad
    &&\psi_{-1,\alpha}^{0,+}=-q_{\alpha,z},\\
    &\psi_{0,\alpha}^{0,-}=r_{\alpha,t},\qquad
    &&\psi_{0,\alpha}^{0,+}=q_{\alpha,t},\\
    &\psi_{1,\alpha}^{0,-}=r_{\alpha-1,t}-u_t r_{\alpha,x},\qquad
    &&\psi_{1,\alpha}^{0,-}=q_{\alpha+1,t}-u_t q_{\alpha,x},
  \end{align*}
  and
  \begin{align*}
     \omega_{-2,\alpha}^{0,-} &=-(\alpha+2)r_{\alpha+2} - zr_{\alpha+2,z} +
     (r_1+zr_{1,z})r_{\alpha+1,y} + (2r_2+zr_{2,z}-(r_1+zr_{1,z})r_{1,y})r_{\alpha,y},\\
   \omega_{-2,\alpha}^{0,+} &=zq_{\alpha-3,t} -zu_z q_{\alpha-3,x} +
   (\alpha-1)q_{\alpha-2} +u_y(r_1+zr_{1,z})q_{\alpha-2,x}\\
   &+u_y(2r_2+zr_{2,z}-(r_1+zr_{1,z})r_{1,y})q_{\alpha-1,x},\\[2pt]
   \omega_{-1,\alpha}^{0,-}&
   =-(\alpha+1)r_{\alpha+1}-zr_{\alpha+1,z}+(r_1+zr_{1,z})r_{\alpha,y},\\ 
   \omega_{-1,\alpha}^{0,+}& =zq_{\alpha-2,t}-zu_z
   q_{\alpha-2,x}+\alpha q_{\alpha-1}+u_y(r_1+zr_{1,z})q_{\alpha-1,x},\\[2mm]
   \omega_{0,\alpha}^{0,-}&=-\alpha r_\alpha-zr_{\alpha,z},\\
   \omega_{0,\alpha}^{0,+}&=(\alpha+1)q_\alpha-zq_{\alpha,z},\\[2pt]
   \omega_{1,\alpha}^{0,-} &=-(\alpha -1)r_{\alpha -1}-zr_{\alpha
     -1,z}+(zr_{1,t}-ur_{1,x})r_{\alpha -1,y},\\  
   \omega_{1,\alpha}^{0,+} &=(\alpha +2)q_{\alpha +1}-uq_{\alpha
     ,x}+zq_{\alpha ,t},\\[2mm]
   \omega_{2,\alpha}^{0,-} &=-(\alpha -2)r_{\alpha -2}-zr_{\alpha
     -2,z}-(u-zu_z)r_{\alpha -1,x}-(2q_1-uu_x+zu_t)r_{\alpha ,x},\\  
   \omega_{2,\alpha}^{0,+} &=(\alpha +3)q_{\alpha +2}-uq_{\alpha
     +1,x}+zq_{\alpha +1,t}-(2q_1-uu_x+zu_t)q_{\alpha ,x}. 
  \end{align*}
  To lift the shadows~$\xi_{i,0}(A)$ and~$\upsilon_{i,0}(B)$, let us introduce
  the operators
  \begin{equation*}
    \mathcal{Y}= -t \pd{}{z} +
    \sum_{i=0}^{\infty}(i+1)r_{i+1}\pd{}{r_i},\qquad 
    \mathcal{X}= -z \pd{}{t} + u \pd{}{x} +
    \sum_{i=0}^{\infty}(i+2)q_{i+1}\pd{}{q_i} 
    \end{equation*}
    (recall that $r_0=y$ and $q_0=u$) and the quantities $P_\alpha(A)$,
    $Q_\alpha(B)$, $j=0,1,2,\dots$, defined by induction as follows:
    \begin{equation*}
      P_0(A)=A, \quad P_\alpha (A)=\frac{1}{\alpha }\mathcal{Y}(P_{\alpha
        -1}(A)),\qquad   Q_0(B)=B, \quad Q_\alpha (B)=\frac{1}{\alpha
      }\mathcal{X}(Q_{\alpha -1}(B)), \qquad \alpha \geq 1. 
    \end{equation*}
    We also tacitly assume that $P_\alpha(A)$ and $Q_\alpha(B)$ vanish if
    $\alpha$ is negative.  Then
    \begin{align*}
      \xi_{-1,\alpha}^-(A)& =A(r_{1,y}r_{\alpha ,y}-r_{\alpha
        +1,y})-A_yr_1r_{\alpha ,y}+A_ztr_{\alpha ,y}+P_{\alpha +1}(A),\\ 
      \xi_{-1,\alpha}^+(A)&=u_y(A(r_{1,y}q_{\alpha -1,x}-q_{\alpha -2,x}) -A_y
      r_1 q_{\alpha -1,x} + A_z tq_{\alpha -1,x}),\\[2pt]
      \xi_{0,\alpha}^-(A)&=-Ar_{\alpha ,y}+P_\alpha (A), \\
      \xi_{0,\alpha}^+(A)&=-Aq_{\alpha,y},\\[2pt]
      \xi_{1,\alpha}^-(A)&=P_{\alpha-1}(A), \\
      \xi_{1,\alpha}^+(A)&=0,
      \intertext{and}
      \upsilon_{-1,\alpha}^-(B)&=0, \\
      \upsilon_{-1,\alpha}^+(B)&=Q_\alpha(B),\\[2pt]
      \upsilon_{0,\alpha}^-(B)&=-Br_{\alpha,x}, \\
      \upsilon_{0,\alpha}^+(B)&=-Bq_{\alpha ,x}+Q_{\alpha +1}(B),\\[2pt]
      \upsilon_{1,\alpha}^-(B)&=B (u_xr_{1,x}r_{\alpha -1,y}- r_{\alpha
          -1,x} ) - B_x u r_{1,x} r_{\alpha -1,y} + B_t z
        r_{1,x}r_{\alpha -1,y},\\ 
        \upsilon_{1,\alpha}^+(B)&=B ( u_xq_{\alpha ,x}-q_{\alpha +1,x})-B_x u
        q_{\alpha ,x}+ B_t zq_{\alpha ,x}+Q_{\alpha +2}(B).  
      \end{align*}
      It is atrightforward to check that these are indeed the needed lifts.
  \end{proof}
  Thus, we obtained fourteen symmetries
  \begin{gather*}
    \Psi_{-1}^0,\quad \Psi_0^0,\quad \Psi_1^0,\\
    \Omega_{-2}^0,\quad \Omega_{-1}^0,\quad \Omega_0^0,\quad \Omega_1^0,\quad
    \Omega_2^0,\\
    \Xi_{-1}(A),\quad \Xi_0(A),\quad \Xi_1(A),\\
    \Upsilon_{-1}(B),\quad \Upsilon_0(B),\quad \Upsilon_1(B)
  \end{gather*}
  in $\tau^\pm$ which will serve as seeds for construction the entire algebra
  of nonlocal symmetries.

\begin{remark}
  It is worth to note that the operators~$\mathcal{X}$ and~$\mathcal{Y}$ used
  in the proof of Proposition~\ref{sec:algebr-nonl-symm-prrop-1} have a
  transparent geometrical interpretation. Namely, consider the system
  consisting of Equations~\eqref{eq:1} and~\eqref{eq:3} and let us treat the
  parameter~$\lambda$ as an additional independent variable with the condition
  $u_\lambda = 0$. Then the total derivative~$D_\lambda$ transforms
  to~$\mathcal{X}$ when passing from the covering~\eqref{eq:3} to~$\tau^+$ and
  to~$\mathcal{Y}$ when passing to~$\tau^-$.
\end{remark}

\textbf{Construction of hierarchies 1.}
Now we use the symmetries~$\Omega_{\pm 1}^0$ as hereditary ones and construct
two infinite hierarchies
\begin{equation*}
  \Psi_i^0 =
  \begin{cases}
    \dfrac{1}{i+1}\{\Omega_{-1}^0,\Psi_{i+1}^0\},&\text{if } i\leq-2,\\[10pt]
    \dfrac{1}{i-1}\{\Omega_1^0,\Psi_{i-1}^0\},&\text{if } i\geq 2,
  \end{cases}
\end{equation*}
and
\begin{equation*}
  \Omega_i^0 =
  \begin{cases}
    \dfrac{1}{i+2}\{\Omega_{-1}^0,\Omega_{i+1}^0\},&\text{if } i\leq-3,\\[10pt]
    \dfrac{1}{i-2}\{\Omega_1^0,\Omega_{i-1}^0\}&\text{if } i\geq 3.
  \end{cases}
\end{equation*}
These hierarchies will be used below to construct new ones.
\medskip

\textbf{Construction of hierarchies 2.} Define the functions
\begin{equation}
  \label{eq:10}
  \psi_{i,0}^j = \sum_{m=0}^j(-1)^m\binom{j}{m}t^{j-m}z^m\psi_{i-m,0}^0,\quad j\geq1.
\end{equation}

\begin{theorem}
  \label{sec:algebr-nonl-symm-thm-1}
  Formula~\eqref{eq:10} defines shadows of symmetries. These shadows can be
  lifted to $\tau^\pm$ and thus define infinite number of
  hierarchies~$\{\Psi_i^j\}$ of nonlocal symmetries.
\end{theorem}

\begin{proof}
  The proof is accomplished in two steps: first we establish
  that~$\psi_{i,0}^j$ are shadows and then show that they can be lifted.

  \emph{Step $1$}. Induction on~$j$. To this end, let us rewrite~\eqref{eq:10}
  recursively. Namely, we write
  \begin{equation}
    \label{eq:14}
    \psi_{i,0}^j = t\psi_{i,0}^{j-1} -z\psi_{i-1,0}^{j-1},\quad i\geq 1.
  \end{equation}
  Let~$j=1$. Consider the linearization operator lifted to $\tau^\pm$
  \begin{equation}
    \label{eq:17}
    \ell_{\mathcal{E}}^\pm = D_y^\pm D_t^\pm - u_{xy}D_z^\pm + u_{xz}D_y^\pm -
    u_zD_x^\pm D_y^\pm + u_yD_x^\pm D_z^\pm.
  \end{equation}
  Then due to~\eqref{eq:17} for $j=1$ one obviously has
  \begin{multline*}
    \ell_{\mathcal{E}}^\pm(\psi_{i,0}^1) =
    \ell_{\mathcal{E}}^\pm(t\psi_{i,0}^0 - z\psi_{i-1,0}^0) =\\
    t\ell_{\mathcal{E}}^\pm(\psi_{i,0}^0) -
    z\ell_{\mathcal{E}}^\pm(\psi_{i-1,0}^0) + D_y^\pm(\psi_{i,0}^0) +
    u_{xy}\psi_{i-1,0}^0 - u_yD_x^\pm(\psi_{i-1,0}^0) =\\  D_y^\pm(\psi_{i,0}^0) +
    u_{xy}\psi_{i-1,0}^0 - u_yD_x^\pm(\psi_{i-1,0}^0),
  \end{multline*}
  since $\psi_{i,0}^0$ and $\psi_{i-1,0}^0$ are shadows. But the last term in
  the equalities above is exactly the first equation in the
  formula~\eqref{eq:16} for the recursion operator. It can be checked that
  this operator, modulo the image of zero (see discussion in the end of the
  paper) connects the shadows~$\psi_{i,0}^0$ and~$\psi_{i-1,0}^0$. In particular,
  \begin{equation}\label{eq:18}
    D_y^\pm(\psi_{i,0}^0) +
    u_{xy}\psi_{i-1,0}^0 - u_yD_x^\pm(\psi_{i-1,0}^0)=0,
  \end{equation}
  because all~$\psi_{\alpha,0}^0$ are shadows. Moreover, since~\eqref{eq:18}
  does not contain the total derivatives in~$z$ and~$t$, we deduce,
  using~\eqref{eq:14},that
  \begin{equation*}
    D_y^\pm(\psi_{i,0}^1) +
    u_{xy}\psi_{i-1,0}^1 - u_yD_x^\pm(\psi_{i-1,0}^1)=0.
  \end{equation*}

  Let now~$j>1$ and assume that for all~$l<j$ and~$i\in\mathbb{Z}$ the
  functions~$\psi_{i,0}^l$ are shadows that enjoy the relations
  $D_y^\pm(\psi_{i,0}^l) + u_{xy}\psi_{i-1,0}^l -
  u_yD_x^\pm(\psi_{i-1,0}^l)=0$. Then the proof of the induction step is
  exactly the same as the one for the case~$j=1$.\medskip
  
  \emph{Step $2$}. We shall now prove that the functions
  \begin{equation}
    \label{eq:19}
    \psi_{i,\alpha}^{j,\pm} = t\psi_{i,\alpha}^{j-1,\pm} - z\psi_{i-1,\alpha}^{j-1,\pm}
  \end{equation}
  satisfy System~\eqref{eq:15} for all~$i\in\mathbb{Z}$, $j\geq0$, $\alpha\geq
  1$.  We also use induction on~$j$ here.

  Consider the case~$j=1$. Substituting the expression $\psi_{i,\alpha}^{1,\pm}
  = t\psi_{i,\alpha}^{0,\pm} - z\psi_{i-1,\alpha}^{0,\pm}$ to the defining
  equations~\eqref{eq:15}, we obtain for~$\tau^-$
  \begin{equation*}
    D_t^-(t\psi_{i,1}^{0,-} - z\psi_{i-1,1}^{0,-}) =
    \frac{1}{u_y^2}\left(u_yD_z^-(t\psi_{i,0}^{0,-} - z\psi_{i-1,0}^{0,-})
      -u_zD_y^-(t\psi_{i,0}^{0,-} - z\psi_{i-1,0}^{0,-})\right)  
  \end{equation*}
  in the case~$\alpha=1$ and
  \begin{multline*}
    D_t^-(t\psi_{i,\alpha}^{0,-} - z\psi_{i-1,\alpha}^{0,-}) =
    \frac{1}{u_y^2}\left(u_yD_z^-(t\psi_{i,0}^{0,-} -
      z\psi_{i-1,0}^{0,-}) - u_zD_y^-(t\psi_{i,0}^{0,-} -
      z\psi_{i-1,0}^{0,-})\right)\\
    + \frac{1}{u_y}\left((u_yD_z^-(t\psi_{i,0}^{0,-} -
      z\psi_{i-1,0}^{0,-}) - u_zD_y^-(t\psi_{i,\alpha-1}^{0,-} -
      z\psi_{i-1,\alpha-1}^{0,-})\right),
  \end{multline*}
  when~$\alpha>1$. But the functions~$\psi_{i,\alpha}^{0,-}$ are the
  components of the nonlocal symmetries~$\Psi_i^0$ and hence we obtain the
  conditions
  \begin{equation}
    \label{eq:20}
    \psi_{i,1}^{0,-} = -1\frac{1}{u_y}\psi_{i-1,0}^{0,-},\qquad
    \psi_{i,\alpha}^{0,-} = -\frac{1}{u_y}\psi_{i-1,0}^{0,-} +
    \psi_{i-1,\alpha-1}^{0,-},\quad \alpha>1.
  \end{equation}
  from the above equations.

  Similar computations show that the conditions
  \begin{equation}
    \label{eq:21}
    \psi_{i-1,1}^{0,+} = u_z\psi_{i-1,0}^{0,+},\qquad \psi_{i-1,\alpha}^{0,+}
    = q_{\alpha-1,x}\psi_{i-1,0}^{0,+} + \psi_{i,\alpha-1}^{0,+},\quad \alpha>1.
  \end{equation}
  must hold in~$\tau^+$.
  
  \begin{lemma}
    \label{sec:algebr-nonl-symm-lemm-1}
    Conditions~\eqref{eq:20} and~\eqref{eq:21} do hold for all~$\alpha>1$
    and~$i\in \mathbb{Z}$.
  \end{lemma}
  \begin{proof}[Proof of Lemma~\textup{\ref{sec:algebr-nonl-symm-lemm-1}}]
    The proof comprises two inductions on~$i$ (for~$i\geq0$ and~$i\leq0$) and
    consists of voluminous computations based on explicit descriptions from
    Proposition~\ref{sec:algebr-nonl-symm-prrop-1} and on the definition of
    the symmetries~$\Psi_i^0$. We omit the details.
  \end{proof}

  Note now that the functions~$\psi_{i,\alpha}^{1,\pm} =
  t\psi_{i,\alpha}^{0,\pm} - z\psi_{i-1,\alpha}^{0,\pm}$ satisfy the
  conditions similar to~\eqref{eq:20} and~\eqref{eq:21} by linearity. This
  finishes the proof of the induction base. The proof of the induction step
  does not differ from the latter.
\end{proof}

In a similar way, we define the functions
\begin{equation}
  \label{eq:11}
  \omega_{i,0}^j =
  \sum_{m=0}^j(-1)^m\binom{j}{m}t^{j-m}z^m\omega_{i-m,0}^0,\quad j\geq1, 
\end{equation}
and prove the following
\begin{theorem}
  \label{sec:algebr-nonl-symm-thm-2}
  Formula~\eqref{eq:11} defines shadows of symmetries. These shadows can be
  lifted to $\tau^\pm$ and thus define infinite number of
  hierarchies~$\{\Omega_i^j\}$ of nonlocal symmetries.
\end{theorem}
\begin{proof}[The proof almost exactly copies the one of
  Theorem~\textup{\ref{sec:algebr-nonl-symm-thm-1}}]  
\end{proof}

\begin{remark}
  As it follows from Theorems~\ref{sec:algebr-nonl-symm-thm-1}
  and~\ref{sec:algebr-nonl-symm-thm-2}, the hierarchies~$\{\Psi_i^j\}$,
  $\{\Omega_i^j\}$, $i\in\mathbb{Z}$, $j\geq0$, exist in the Whitney
  product~$\tau^\pm$, but this result may be clarified. More detailed
  information on the $\Psi$-hierarchies is presented in
  Table~\ref{tab:psy}. Note that the symmetries~$\Psi_{-1}^0$, $\Psi_0^0$,
  and~$\Psi_0^1$ are local.
  Additional properties of the $\Omega$-hierarchies are given in
  Table~\ref{tab:omega}. Of all these symmetries, only~$\Omega_0^0$ is a local
  one.
  \begin{table}[h]
    \centering
    \begin{tabular}{||c||c|c||}\hline\hline
      $\Psi_i^j$&$j<i+2$&$j\geq i+2$\\\hline\hline
      $i\leq0$&in $\tau^-$, $\tau^+$, $\tau^\pm$&in $\tau^-$, $\tau^\pm$\\\hline
      $i>0$&in $\tau^+$, $\tau^\pm$&in $\tau^\pm$ only\\\hline\hline
    \end{tabular}\smallskip
    \caption{Distribution of $\Psi_i^j$ over $\tau^-$, $\tau^+$, and $\tau^\pm$}
    \label{tab:psy}
  \end{table}  
  \begin{table}[h]
    \centering
    \begin{tabular}{||c||c|c||}\hline\hline
      $\Omega_i^j$&$j<i+1$&$j\geq i+1$\\\hline\hline
      $i\leq0$&in $\tau^-$, $\tau^+$, $\tau^\pm$&in $\tau^-$, $\tau^\pm$\\\hline
      $i>0$&in $\tau^+$, $\tau^\pm$&in $\tau^\pm$ only\\\hline\hline
    \end{tabular}\smallskip
    \caption{Distribution of $\Omega_i^j$ over $\tau^-$, $\tau^+$, and
      $\tau^\pm$} 
    \label{tab:omega}
  \end{table}
\end{remark}

\textbf{Construction of hierarchies 3.} The last step is the construction of
the $(x,t)$- and $(y,x)$-dependent hierarchies. To this end, we set
\begin{equation*}
  \Xi_i(A) =
  \begin{cases}
    \dfrac{1}{i+1}\{\Omega_{-1}^0+\Psi_{-1}^1,\Xi_{i+1}(A)\},&\text{if }i\leq
    -2,\\[10pt] 
    \dfrac{1}{i-1}\{\Omega_1^0+\Psi_1^1,\Xi_{i-1}(A)\},&\text{ if }i\geq2,
  \end{cases}
\end{equation*}
and
\begin{equation*}
  \Upsilon_i(B) =
  \begin{cases}
    \dfrac{1}{i+1}\{\Omega_{-1}^0,\Upsilon_{i+1}(B)\},&\text{if }i\leq-2,\\[10pt]
    \dfrac{1}{i-1}\{\Omega_1^0,\Upsilon_{i-1}(B)\}&\text{if }i\geq2
  \end{cases}
\end{equation*}
(recall that $A=A(y,z)$ and~$B=B(x,t)$ are arbitrary smooth functions).

\begin{remark}
  As above, the structure of these hierarchies may be clarified in some
  respects. Namely, we have the following facts:
  \begin{equation*}
    \Xi_i(A)\text{ is a symmetry in }
    \begin{cases}
      \tau^-,\tau^\pm,&\text{if }i\leq-1,\\
      \tau^-,\tau^+,\tau^\pm,&\text{if }i=0,\\
      \tau^\pm,&\text{if }i\geq1.
    \end{cases}
  \end{equation*}
  Moreover, the symmetry~$\Xi_0(A)$ is local, while~$\Xi_i(A)$ are invisible
  symmetries for~$i\geq1$.

  In a similar way,
  \begin{equation*}
    \Upsilon_i(B)\text{ is a symmetry in }
    \begin{cases}
      \tau^\pm,&\text{if }i\leq -2,\\
      \tau^-,\tau^+,\tau^\pm,&\text{if }i=-1,0,\\
      \tau^+,\tau^\pm,&\text{if }i\geq1.
    \end{cases}
  \end{equation*}
  The symmetries~$\Upsilon_i(B)$ are invisible for all~$i\leq -2$ and the
  symmetries~$\Upsilon_{-1}(B)$, $\Upsilon_0(B)$ are local ones.
\end{remark}
\medskip

\textbf{Lie algebra structure.} Let us now describe the structure of the Lie
algebra formed by the above constructed symmetries. To this end, relabel some
of them to make the results look neater. Namely, we change notation as
follows:
\begin{equation*}
  \Psi_i^j\mapsto-\Psi_i^{j+1},\qquad \Xi_i(A)\mapsto\Xi_i(A\cdot z^{-i}).
\end{equation*}
Then we have the following result:
\begin{theorem}
  \label{sec:algebr-nonl-symm-thm-3}
  The Lie algebra~$\mathfrak{g} = \sym_{\tau^\pm}(\mathcal{E})$ of the
  $\tau^\pm$-nonlocal symmetries for the 4D MASh equation as an
  $\mathbb{R}$-vector space is generated by the elements
  \begin{equation*}
    \{\Psi_i^j\}_{i\in\mathbb{Z}}^{j\geq1},\quad
    \{\Omega_i^j\}_{i\in\mathbb{Z}}^{j\geq0},\quad
    \{\Upsilon_i(B)\}_{i\in\mathbb{Z}},\quad \{\Xi_i(A)\}_{i\in\mathbb{Z}},
  \end{equation*}
  where~$B=B(x,t)$ and $A=A(y,z)$ are arbitrary smooth functions. They enjoy the
  commutator relations presented in Table~\textup{\ref{tab:Lie-alg}}.
  \begin{table}[h]
    \centering
    \begin{tabular}{||c||c|c|c|c||}\hline\hline
      &$\Psi_k^l$
      &$\Omega_k^l$
      &$\Upsilon_k(\bar{B})$
      &$\Xi_k(\bar{A})$\\\hline\hline
      $\Psi_i^j$
      &$(l-j)\Psi_{i+k}^{j+l}$
      &$l\Omega_{i+k}^{j+l} - i\Psi_{i+k}^{j+l}$
      &$\Upsilon_{i+k}(t^{j+1}\bar{B}_t)$
      &$(-1)^j\Xi_{k+i-j}(z^{i+1}\bar{A}_z-kz^i\bar{A})$\\\hline
      $\Omega_i^j$&
      &$(k-i)\Omega_{i+k}^{j+l}$
      &$\Upsilon_{i+k}(kt^j\bar{B})$
      &$(-1)^j\Xi_{k+i-j}(z^{i+1}\bar{A}_z)$\\\hline
      $\Upsilon_i(B))$&
      &
      &$\Upsilon_{i+k}([B,\bar{B}])$&0\\\hline
      $\Xi_i(A)$&
      &
      &&$\Xi_{i+k}([A,\bar{A}])$\\\hline\hline
    \end{tabular}\smallskip
    \caption{The Lie algebra structure}
    \label{tab:Lie-alg}
  \end{table}\\  
  Here the notation
  \begin{equation*}
    [A,\bar{A}] = A\bar{A}_y-\bar{A}A_y,\qquad [B,\bar{B}] = B\bar{B}_x -
    \bar{B}B_x 
  \end{equation*}
  was used.
\end{theorem}
\begin{proof}
  The proof is omitted due to its extreme length. It consists of a number of
  inductions with explicit computations in the bases of these inductions. 
\end{proof}

\begin{remark}
  Denote by~$\mathfrak{h}\subset\mathfrak{g}$ the subalgebra spanned by the
  elements~$\Psi_i^j$, $\Omega_i^j$, and~$\Upsilon_i(B)$, and
  let~$\mathfrak{i}(A)\subset\mathfrak{g}$ denote the
  ideal~$\{\Xi_i(A)\}$. Then~$\mathfrak{g}$ is the semi-direct product
  $\mathfrak{h}\ltimes\mathfrak{i}(A)$. In its turn, $\mathfrak{h} =
  \mathfrak{h}_0\ltimes \mathfrak{i}(B)$, where
  \begin{equation*}
    \mathfrak{h}_0=\{\Psi_i^j,\Omega_i^j\},\qquad
    \mathfrak{i}(B)=\{\Upsilon_i(B)\} .
  \end{equation*}
  The structure of~$\mathfrak{h}$ is quite clear. Consider the correspondence
  \begin{equation*}
    \Psi_i^j\mapsto t^{j+1}z^i\pd{}{t},\ j\geq1\quad \Omega_i^j\mapsto
    t^jz^{i+1}\pd{}{z},\ j\geq0\quad \Upsilon_i(B)\mapsto z^iB\pd{}{y},\qquad
    i\in \mathbb{Z}.
  \end{equation*}
  Then we obtain an isomorphism between~$\mathfrak{h}$ and the Lie algebra of
  the corresponding vector fields. The action of~$\mathfrak{h}$
  on~$\mathfrak{i}(A)$ is less conventional (see the last column of
  Table~\ref{tab:Lie-alg}).
\end{remark}

\textbf{Action of the recursion operator.} Let us now describe the action of
the recursion operator~\eqref{eq:16} on the shadows our symmetries. First of
all note that
\begin{equation*}
  \overrightarrow{\mathcal{R}}(0) = \xi_{0,0}(A),\qquad
  \overleftarrow{\mathcal{R}}(0) = \upsilon_{0,0}(B),
\end{equation*}
and thus the action is defined modulo the images of zero. Keeping this in mind
we have
\begin{equation*}
  \xymatrix{
    \dots
    \ar@/^/[r]^-{\overleftarrow{\mathcal{R}}}
    &\psi_{0,-2}^j
    \ar@/^/[l]^-{\overrightarrow{\mathcal{R}}}
    \ar@/^/[r]^-{\overleftarrow{\mathcal{R}}}
    &\psi_{0,-1}^j
    \ar@/^/[l]^-{\overrightarrow{\mathcal{R}}}
    \ar@/^/[r]^-{\overleftarrow{\mathcal{R}}}
    &\psi_{0,0}^j
    \ar@/^/[l]^-{\overrightarrow{\mathcal{R}}}
    \ar@/^/[r]^-{\overleftarrow{\mathcal{R}}}
    &\psi_{0,1}^j
    \ar@/^/[l]^-{\overrightarrow{\mathcal{R}}}
    \ar@/^/[r]^-{\overleftarrow{\mathcal{R}}}
    &\psi_{0,2}^j
    \ar@/^/[l]^-{\overrightarrow{\mathcal{R}}}
    \ar@/^/[r]^-{\overleftarrow{\mathcal{R}}}
    &\dots
    \ar@/^/[l]^-{\overrightarrow{\mathcal{R}}}
    \\
    \dots
    \ar@/^/[r]^-{\overleftarrow{\mathcal{R}}}
    &\omega_{0,-2}^j
    \ar@/^/[l]^-{\overrightarrow{\mathcal{R}}}
    \ar@/^/[r]^-{\overleftarrow{\mathcal{R}}}
    &\omega_{0,-1}^j
    \ar@/^/[l]^-{\overrightarrow{\mathcal{R}}}
    \ar@/^/[r]^-{\overleftarrow{\mathcal{R}}}
    &\omega_{0,0}^j
    \ar@/^/[l]^-{\overrightarrow{\mathcal{R}}}
    \ar@/^/[r]^-{\overleftarrow{\mathcal{R}}}
    &\omega_{0,1}^j
    \ar@/^/[l]^-{\overrightarrow{\mathcal{R}}}
    \ar@/^/[r]^-{\overleftarrow{\mathcal{R}}}
    &\omega_{0,2}^j
    \ar@/^/[l]^-{\overrightarrow{\mathcal{R}}}
    \ar@/^/[r]^-{\overleftarrow{\mathcal{R}}}
    &\dots
    \ar@/^/[l]^-{\overrightarrow{\mathcal{R}}}\\
    \dots
    \ar@/^/[r]^-{\overrightarrow{\mathcal{R}}}
    &\upsilon_{1,0}(B)
    \ar@/^/[r]^-{\overrightarrow{\mathcal{R}}}
    \ar@/^/[l]^-{\overleftarrow{\mathcal{R}}}
    &\upsilon_{0,0}(B)
    \ar@/^/[r]^-{\overrightarrow{\mathcal{R}}}
    \ar@/^/[l]^-{\overleftarrow{\mathcal{R}}}
    &0
    \ar@/^/[r]^-{\overrightarrow{\mathcal{R}}}
    \ar@/^/[l]^-{\overleftarrow{\mathcal{R}}}
    &\xi_{0,0}(A)
    \ar@/^/[r]^-{\overrightarrow{\mathcal{R}}}
    \ar@/^/[l]^-{\overleftarrow{\mathcal{R}}}
    &\xi_{1,0}(A)
    \ar@/^/[r]^-{\overrightarrow{\mathcal{R}}}
    \ar@/^/[l]^-{\overleftarrow{\mathcal{R}}}
    &\dots
    \ar@/^/[l]^-{\overleftarrow{\mathcal{R}}}
  }
\end{equation*}
\begin{remark}
  To conclude, recall that in~\cite{Mor-Ser} an infinite series of pair-wise
  commuting nonlocal symmetries was presented. The
  algebra~$\sym_{\tau^\pm}(\mathcal{E})$ described above contains infinite
  number of such hierarchies. Namely, for any~$i\in\mathbb{Z}$ and~$j\in
  \mathbb{N}$ each of the families
  \begin{equation*}
    \mathbf{\Psi}^j = \{\Psi_i^j\}_{i\in\mathbb{Z}}\text{ and }\mathbf{\Omega}_i =
    \{\Omega_i^j\}^{j\geq0} 
  \end{equation*}
  consists of pair-wise commuting symmetries. In addition, if we fix the
  functions~$A(y,z)$ and~$B(x,t)$ then the families
  \begin{equation*}
    \mathbf{\Xi}(A) = \{\Xi_i(A)\}_{i\in\mathbb{Z}}\text{ and }
    \mathbf{\Upsilon}(B) = \{\Upsilon_i(B)\}_{i\in\mathbb{Z}}
  \end{equation*}
  will possess the same property.
\end{remark}

\section*{Acknowledgments}
\label{sec:acknowledgments}

Computations were supported by the \textsc{Jets} software,~\cite{Jets}. The
work of I.K.\ was partially supported by Russian Foundation for Basic Research
Grant 18-29-10013 and Simons-IUM Fellowship Grant 2020.

\end{document}